\documentclass[11pt]{article}%
\usepackage{amsfonts}
\usepackage{amsmath}
\usepackage{amssymb}
\usepackage{graphicx}%
\providecommand{\U}[1]{\protect\rule{.1in}{.1in}}
\newtheorem{theorem}{Theorem}

\newtheorem{proposition}[theorem]{Proposition}
\newtheorem{remark}[theorem]{Remark}

\newenvironment{proof}[1][Proof]{\noindent\textbf{#1.} }{\ \rule{0.5em}{0.5em}}
\begin{document}

\title{A Few Almost Trivial Notes on the Symplectic Radon Transform and the
Tomographic Picture of Quantum Mechanics}
\author{Maurice A. de Gosson\thanks{maurice.de.gosson@univie.ac.at}\\University of Vienna\\Faculty of Mathematics (NuHAG)}
\maketitle

\begin{abstract}
We emphasize in these pedagogical notes the that the theory of the Radon
transform and its applications is best understood using the theory of the
metaplectic group and the quadratic Fourier transforms generating metaplectic
operator.. Doing this we hope that these notes will be useful to a larger
audience, including researchers in time-frequency analysis.

\end{abstract}

\section{Introduction}

In many texts\footnote{A short non-exhaustive list is
\cite{Asorey,Facchi,dariano,damapa,Ibort,Ibortbis}} studying the tomographic
picture of quantum mechanics one finds the following definition of the Radon
transform of a quantum state $\widehat{\rho}$:
\begin{equation}
R_{\widehat{\rho}}(X,\mu,\nu)=\int W(x,p)\delta(X-\mu x-\nu p)dpdx\label{1}%
\end{equation}
where $\mu$ and $\nu$ are real numbers and $W(x,p)$ is the Wigner distribution
of $\widehat{\rho}$:%
\begin{equation}
W(x,p)=\frac{1}{2\pi\hbar}\int_{-\infty}^{\infty}e^{-\frac{i}{\hbar}py}\langle
x+\tfrac{1}{2}y|\widehat{\rho}|x-\tfrac{1}{2}y\rangle dy.\label{2}%
\end{equation}
One then also finds the following expression of the \textquotedblleft inverse
Radon transform\textquotedblright:%
\begin{equation}
W\psi(x,p)=\frac{1}{2\pi\hbar}\int R_{\widehat{\rho}}(X,\mu,\nu)e^{\frac
{i}{\hslash}(X-\mu x-\nu p)}dXd\mu d\nu\label{3}%
\end{equation}

The expression (\ref{1}), which should be interpreted as a distributional
bracket to have any meaning, is however difficult to justify mathematically.
It only makes sense for a restricted class of functions $Wx,p)$ (this function
has to be at least continuous, but this assumption is not necessarily made in
the usual texts). Similarly, the \textquotedblleft
derivation\textquotedblright\ of the inversion formula (\ref{3}) is usually
obscure and does not discuss convergence issues. In this short Note we propose
a rigorous redefinition of the Radon transform which has the additional
advantage of replacing this notion where it belongs, namely rigorous harmonic
analysis and the theory of the metaplectic group and its extensions. (We use
the definition of the metaplectic group using quadratic Fourier transforms
shortly reviewed in APPENDIX\ A.)

\section{The Radon transform}

We will only consider the case of a \textquotedblleft pure
state\textquotedblright, represented by function $\psi\in L^{2}(\mathbb{R})$.
The general case is then easy to obtain since by the spectral theorem for
trace class operators the Wigner distribution of an arbitrary state $\rho$ is
a convex sum of Wigner functions $W\psi_{j}$ \cite{QHA}.  

\begin{proposition}
Let $\widehat{\rho}$ be a pure state $|\psi\rangle$ with $\psi\in
L^{2}(\mathbb{R})$. We assume that $\psi\in L^{1}(\mathbb{R})\cap
L^{2}(\mathbb{R})$ and $\widehat{\psi}\in L^{1}(\mathbb{R})\cap L^{2}%
(\mathbb{R})$ ($\widehat{\psi}$ the Fourier transform of $\psi$). (i) The
(symplectic) Radon transform of the Wigner function $W\psi(x,p)$ is given by
\begin{equation}
R_{\widehat{\rho}}(X,\mu,\nu)=\lambda^{-1}|\widehat{U}_{\mu,\nu}\psi
(\lambda^{-1}X)|^{2}.\label{WU}%
\end{equation}
where $\widehat{U}_{\mu,\nu}\in\operatorname*{Mp}(n)$ is anyone of the two
metaplectic operators covering the rotation $U_{\mu,\nu}=%
\begin{pmatrix}
\mu/\lambda & \nu/\lambda\\
-\nu/\lambda & \mu/\lambda
\end{pmatrix}
$ where $\lambda=\sqrt{\mu^{2}+\nu^{2}}$. (ii) The inverse Radon transform is
given by the formula:%
\begin{equation}
W\psi(x,p)=\frac{1}{2\pi\hbar}\int R_{\widehat{\rho}}(X,\mu,\nu)e^{\frac
{i}{\hslash}(X-\mu x-\nu p)}dXd\mu d\nu\label{inverse}%
\end{equation}
understood as a Fourier transform of the function $R_{\widehat{\rho}}%
(X,\mu,\nu)$.
\end{proposition}

\begin{proof}
(i) Let us make the change of variables
\begin{equation}%
\begin{pmatrix}
u\\
v
\end{pmatrix}
=%
\begin{pmatrix}
\mu/\lambda & \nu/\lambda\\
-\nu/\lambda & \mu/\lambda
\end{pmatrix}%
\begin{pmatrix}
x\\
p
\end{pmatrix}
\label{unitary}%
\end{equation}
in the integral (\ref{1}).This leads to the expression%
\begin{equation}
R_{\widehat{\rho}}(X,\mu,\nu)=\int W\psi(U_{\mu,\nu}^{-1}(u,v))\delta
(X-\lambda u)dudv.\label{dirac}%
\end{equation}
Since $\delta(X-\lambda u)=\lambda^{-1}\delta(\lambda^{-1}X-u)$ this can be
rewritten
\begin{equation}
R_{\widehat{\rho}}(X,\mu,\nu)=\lambda^{-1}\int W\psi(\lambda^{-1}(\mu u-\nu
v,\nu u+\mu v))\delta(\lambda^{-1}X-u)dudv\label{diracbis}%
\end{equation}
In view of the symplectic covariance property \cite{Birk,Wigner,Littlejohn} of
the Wigner transform we have%
\begin{equation}
W\psi(U_{\mu,\nu}^{-1}(u,v))=W(\widehat{U}_{\mu,\nu}\psi)(u,v)\label{sympco}%
\end{equation}
where $\widehat{U}_{\mu,\nu}$ is anyone of the two metaplectic operators (see
the APPENDIX A) covering $U$ and hence (\ref{diracbis}) yields%
\begin{align*}
R_{\widehat{\rho}}(X,\mu,\nu) &  =\lambda^{-1}\int W(\widehat{U}_{\mu,\nu}%
\psi)(\lambda^{-1}X,v)\delta(\lambda^{-1}X-u)dudv\\
&  =\lambda^{-1}\int W(\widehat{U}_{\mu,\nu}\psi)(\lambda^{-1}X,v)dv
\end{align*}
hence (\ref{WU}) using the marginal properties
\begin{equation}
\int W\psi(x,p)dp=|\psi(x)|^{2}\text{ \ , \ }\int W\psi(x,p)dx=|\widehat{\psi
}(p)|^{2}\label{marginal}%
\end{equation}
of the Wigner transform, which are valid \cite{Birk,Wigner} since $\psi\in
L^{1}(\mathbb{R})\cap L^{2}(\mathbb{R})$ and $\widehat{\psi}\in L^{1}%
(\mathbb{R})\cap L^{2}(\mathbb{R})$. (ii) Let us denote $A$ the right-hand
side of the equality (\ref{inverse}). Using the first marginal property
(\ref{marginal}) we have
\begin{align*}
A &  =\lambda^{-1}\frac{1}{2\pi\hbar}\int|\widehat{U}_{\mu,\nu}\psi
(\lambda^{-1}X)|^{2}e^{\frac{i}{\hslash}(X-\mu x-\nu p)}dXd\mu d\nu\\
&  =\lambda^{-1}\frac{1}{2\pi\hbar}\int W(\widehat{U}_{\mu,\nu}\psi
)(\lambda^{-1}X,P)e^{\frac{i}{\hslash}(X-\mu x-\nu p)}dXdPd\mu d\nu.
\end{align*}
Replacing $X$ with $\lambda X$ and using the symplectic covariance property
(\ref{sympco}) we get
\begin{align*}
A &  =\frac{1}{2\pi\hbar}\int W(\widehat{U}_{\mu,\nu}\psi)(X,P)e^{\frac
{i}{\hslash}(\lambda X-\mu x-\nu p)}dXdPd\mu d\nu\\
&  =\frac{1}{2\pi\hbar}\int W\psi(U_{\mu,\nu}^{-1}(X,P))e^{\frac{i}{\hslash
}(\lambda X-\mu x-\nu p)}dXdPd\mu d\nu\\
&  =\frac{1}{2\pi\hbar}\int W\psi((\mu/\lambda)X-(\nu/\lambda)P,(\nu
/\lambda)X+(\mu/\lambda)P))e^{\frac{i}{\hslash}(\lambda X-\mu x-\nu
p)}dXdPd\mu d\nu.
\end{align*}
Setting $Y=(\mu/\lambda)X-(\nu/\lambda)P$ and $Z=(\nu/\lambda)X+(\mu
/\lambda)P$ (and hence $\lambda X=\mu Y+\nu Z$) this is
\[
A=\frac{1}{2\pi\hbar}\int W\psi(Y,Z))e^{\frac{i}{\hslash}(\mu(Y-x)+\nu
(Z-p))}dXdPd\mu d\nu.
\]
In view of the Fourier inversion formula, written formally as
\[
\int e^{\frac{i}{\hslash}(\mu(Y-x)+\nu(Z-p))}d\mu d\nu=2\pi\hbar
\delta(Y-x,Z-p)
\]
we thus have
\[
A=\int W\psi(x,p)\delta(Y-x,Z-p)dXdP=W\psi(x,p))
\]
which was to be proven.
\end{proof}

\section{Some explicit calculations}

Assume $\nu\neq0$. We have the following explicit form for the metaplectic
operator $\widehat{U}_{\mu,\nu}$ \cite{Birk,Folland}:%
\begin{equation}
\widehat{U}_{\mu,\nu}\psi(x)=\frac{e^{-i\pi/4}}{\sqrt{2\pi\hbar}}\sqrt
{\frac{\lambda}{\nu}}\int_{-\infty}^{\infty}\exp\left[  \frac{i}{\hbar}\left(
\frac{\mu}{2\nu}x^{2}-\frac{\lambda}{\nu}xx^{\prime}+\frac{\mu}{2\nu}%
x^{\prime2}\right)  \right]  \psi(x^{\prime})dx^{\prime}\label{meta1}%
\end{equation}
where the argument of $\sqrt{\lambda/\nu}$ can take two possible values. It
follows that%
\[
|\widehat{U}_{\mu,\nu}\psi(x)|^{2}=\frac{\lambda}{2\pi\hbar|\nu|}\left\vert
\int_{-\infty}^{\infty}\exp\left[  \frac{i}{\hbar}\left(  -\frac{\lambda}{\nu
}xx^{\prime}+\frac{\mu}{2\nu}x^{\prime2}\right)  \right]  \psi(x^{\prime
})dx^{\prime}\right\vert ^{2}%
\]
so that, taking (\ref{WU}) into account,%
\begin{equation}
R_{\widehat{\rho}}(X,\mu,\nu)=\frac{1}{2\pi\hbar|\nu|}\left\vert \int%
_{-\infty}^{\infty}\exp\left(  -\frac{i}{\hbar\nu}Xx^{\prime}\right)
\exp\left(  \frac{i\mu}{2\hbar\nu}x^{\prime2}\right)  \psi(x^{\prime
})dx^{\prime}\right\vert ^{2}.\label{explicit}%
\end{equation}
which we can rewrite as a Fourier transform%
\begin{equation}
R_{\widehat{\rho}}(X,\mu,\nu)=\frac{1}{|\nu|}\left\vert \widehat{F}\left[
\exp\left(  \frac{i\mu}{2\hbar\nu}x^{\prime2}\right)  \psi\right]  \left(
\frac{X}{\nu}\right)  \right\vert ^{2}.\label{explicitbis}%
\end{equation}
where
\[
\widehat{F}\phi(x)=\frac{e^{-i\pi/4}}{\sqrt{2\pi\hbar}}\int_{-\infty}^{\infty
}e^{-\frac{i}{\hbar}xx^{\prime}}\phi(x^{\prime})dx^{\prime}.
\]

\begin{remark}
In the language of time-frequency analysis one would say that the Radon
transform of $\psi$ is (up to a rescaling factor) the squared modulus of the
Fourier transform of the product of $\psi$ by a \textquotedblleft
chirp\textquotedblright.
\end{remark}

Denoting by $\widehat{V}_{-\mu/\nu}$ the operator of multiplication by
$\exp(i(\mu/\nu)/2\hbar)$ and by $\widehat{M}_{\nu}$ the scaling operator
\[
\widehat{M}_{\nu}\phi(x)=\sqrt{\nu}\phi(\nu x)
\]
this formula can be rewritten%
\[
\widehat{U}_{\mu,\nu}\psi(x)=\widehat{V}_{-\mu/\nu}\widehat{F}(\widehat{M}%
_{\lambda/\mu}\widehat{V}_{-\mu/\nu}\psi)
\]
where $\widehat{F}$ is essentially a Fourier transform: Formula
(\ref{explicitbis}) then becomes
\begin{equation}
R_{\widehat{\rho}}(X,\mu,\nu)=\left\vert \widehat{M}_{\nu}\widehat{F}%
(\widehat{V}_{-\mu/\nu}\psi)(X)\right\vert ^{2}.\label{WUbis}%
\end{equation}

\section{The Pauli problem for Gaussians}

That the full Radon transform of a state is not necessary to reconstruct it is
seen in the elementary example considered in this section, which is a version
of the so-called \textquotedblleft Pauli problem\textquotedblright.
Historically, this problem goes back to the famous question Pauli asked in
\cite{Pauli}, whether the probability densities $|\psi(x)|^{2}$ and
$|\widehat{\psi}(p)|^{2}$ uniquely determine the wavefunction $\psi(x)$. The
answer is negative, as is seen on the following simple example  \cite{mdpi}:
of Gaussian wavefunction in one spatial dimension
\begin{equation}
\psi(x)=\left(  \tfrac{1}{2\pi\sigma_{xx}}\right)  ^{1/4}e^{-\frac{x^{2}%
}{4\sigma_{xx}}}e^{\frac{i\sigma_{xp}}{2\hbar\sigma_{xx}}x^{2}}\label{Gauss1}%
\end{equation}
where $\sigma_{xx}$ is the variance in the position variable and $\sigma_{xp}$
the covariance in the position and momentum variables. Using the formula
giving the Fourier transform $\widehat{\psi}_{a,b}$ of
\begin{equation}
\psi_{a,b}(x)=e^{-\frac{1}{\hbar}(a+ib)x^{2}}\text{ \ , \ }a>0,b\in
\mathbb{R}\label{psiab}%
\end{equation}
which is%
\begin{equation}
\widehat{\psi}_{a,b}(p)=\frac{1}{\sqrt{2\pi\hbar}}\int_{-\infty}^{\infty
}e^{-\frac{i}{\hbar}px}\psi_{a,b}(x)dx=\frac{1}{\sqrt{2(a+ib)}}e^{-\frac
{1}{\hbar}(a+ib)^{-1}p^{2}}\label{FT}%
\end{equation}
we see that the Fourier transform of $\psi$ is explicitly given by
\begin{equation}
\widehat{\psi}(p)=e^{i\gamma}\left(  \tfrac{1}{2\pi\sigma_{pp}}\right)
^{1/4}e^{-\frac{p^{2}}{4\sigma_{pp}}}e^{-\frac{i\sigma_{xp}}{2\hbar\sigma
_{pp}}p^{2}}\label{FGauss1}%
\end{equation}
where $\gamma$ is an unessential constant real phase depending on the
covariances. It follows that we have
\[
|\psi(x)|^{2}=\left(  \tfrac{1}{2\pi\sigma_{xx}}\right)  ^{1/2}e^{-\frac
{x^{2}}{2\sigma_{xx}}}\text{ \ , \ }|\widehat{\psi}(p)|^{2}=\left(  \tfrac
{1}{2\pi\sigma_{pp}}\right)  ^{1/2}e^{-\frac{p^{2}}{2\sigma_{pp}}}%
\]
and these relations imply the knowledge of $\sigma_{xx}$ and of $\sigma_{pp}$.
The covariance $\sigma_{xp}\ $is then determined \emph{up to a sign} because
the state $\psi$ saturates the Robertson--Schr\"{o}dinger inequality: we have
\begin{equation}
\sigma_{xx}\sigma_{pp}-\sigma_{xp}^{2}=\tfrac{1}{4}\hbar^{2}\label{RSUP1}%
\end{equation}
and this identity can be solved in $\sigma_{xp}$ yielding $\sigma_{xp}%
=\pm(\sigma_{xx}\sigma_{pp}-\frac{1}{4}\hbar^{2})^{1/2}$.

Let is now apply the Radon transform in its form (\ref{explicit}) to $\psi$.
For this it suffices to observe that the Fourier formula (\ref{FT}) implies
that%
\begin{equation}
|\widehat{\psi}_{a,b}(p)|^{2}=\frac{1}{2\sqrt{a^{2}+b^{2}}}\exp\left[  \left(
-\frac{2}{\hbar}\frac{a}{a^{2}+b^{2}}\right)  p^{2}\right]  . \label{FTmod}%
\end{equation}
Now, formula (\ref{explicit}) reads here
\begin{multline}
R_{\widehat{\rho}}(X,\mu,\nu)=\nonumber\\
C\left\vert \int_{-\infty}^{\infty}\exp\left(  -\frac{i}{\hbar\nu}Xx^{\prime
}\right)  \exp\left[  i\left(  \frac{\mu}{2\hbar\nu}+\frac{i\sigma_{xp}%
}{2\hbar\sigma_{xx}}\right)  x^{\prime2}\right]  dx^{\prime}\right\vert ^{2}.
\end{multline}
where
\[
C=\frac{1}{2\pi\hbar|\nu|}\left(  \frac{1}{2\pi\sigma_{xx}}\right)  ^{1/2}%
\]
hence, choosing
\[
a=\frac{\hbar}{4\sigma_{xx}}\text{ \ },\text{ \ }b=\frac{\mu}{2\nu}%
+\frac{\sigma_{xp}}{2\sigma_{xx}}.
\]
in (\ref{FTmod}), we have, using the identity $\sigma_{xx}\sigma_{pp}%
+\sigma_{xp}^{2}=\frac{1}{4}\hbar^{2}$,
\[
a^{2}+b^{2}=\frac{\sigma_{pp}}{4\sigma_{xx}}+\left(  \frac{\mu}{2\nu}\right)
^{2}+\frac{\mu}{2\nu}\frac{\sigma_{xp}}{\sigma_{xx}}%
\]
and thus (\ref{explicit}) yields%
\[
R_{\widehat{\rho}}(X,\mu,\nu)=\frac{1}{2\pi\hbar|\nu|}\left\vert \int%
_{-\infty}^{\infty}e^{-\frac{i}{\hbar\nu}Xx^{\prime}}\psi_{a,b}(x^{\prime
})dx^{\prime}\right\vert ^{2}=\frac{1}{|\nu|}\left\vert \widehat{\psi}%
_{a,b}\left(  \frac{X}{\nu}\right)  \right\vert ^{2}%
\]
that is, by (\ref{FTmod}),
\begin{equation}
R_{\widehat{\rho}}(X,\mu,\nu)=\frac{1}{2|\nu|(a^{2}+b^{2})}\exp\left[
-\frac{1}{\hbar}\frac{a}{a^{2}+b^{2}}\left(  \frac{X}{\nu}\right)
^{2}\right]  \label{wouro}%
\end{equation}
where%
\begin{equation}
\frac{a}{a^{2}+b^{2}}=\frac{\hbar}{\sigma_{pp}+4\varepsilon^{2}\sigma
_{xx}+4\varepsilon\sigma_{xp}}\text{ \ },\text{ \ }\varepsilon=\frac{\mu}%
{2\nu}. \label{ab}%
\end{equation}
This formula allows the unambiguous determination of the covariance
$\sigma_{xp}$ once the variances $\sigma_{xx}$ and $\sigma_{pp}$ are known.
Notice that it suffices with one choice of parameters $\mu,\nu$ to determine
the unknown covariance. A geometric explanation will be given below.

\section{Geometric interpretation}

The Wigner transform of the Gaussian (\ref{Gauss1}) is
\cite{Birk,Wigner,Littlejohn}%
\[
W\psi(x,p)=\frac{1}{\pi\hbar}e^{-\frac{1}{\hbar}Gx^{2}}%
\]
where $G=\frac{2}{\hbar}%
\begin{pmatrix}
\sigma_{pp} & -\sigma_{xp}\\
-\sigma_{xp} & \sigma_{xx}%
\end{pmatrix}
$. The matrix
\[
\Sigma=\frac{\hbar}{2}G^{-1}=%
\begin{pmatrix}
\sigma_{xx} & \sigma_{xp}\\
\sigma_{xp} & \sigma_{pp}%
\end{pmatrix}
\]
is thus the usual covariance matrix of the state $\psi$, and to it one
associates the covariance ellipse
\[
\Omega:\frac{1}{2}(x,p)\Sigma^{-1}\tbinom{x}{p}\leq1
\]
that is, explicitly,%
\[
\Omega:\frac{\sigma_{pp}}{2D}x^{2}-\frac{\sigma_{xp}}{D}px+\frac{\sigma_{xx}%
}{2D}p^{2}\leq1
\]
where
\[
D=\sigma_{xx}\sigma_{pp}-\sigma_{xp}^{2}=\tfrac{1}{4}\hbar^{2}.
\]
Consider next the straight lines $\ell_{\mu,\nu}$ in the phase plane defined
by $\mu x+\nu p=0$ and their intersections with $\Omega$. The set $\Omega
\cap\ell_{0,1}$ is the intersection of $\Omega$ with the $x$-axis and is thus
the real interval $[-\hbar/\sqrt{2\sigma_{xx}},\hbar/\sqrt{2\sigma_{xx}}]$
and, similarly, $\Omega\cap\ell_{1,0}$ is the intersection of $\Omega$ with
the $p$-axis, that is, $[-\hbar/\sqrt{2\sigma_{pp}},\hbar/\sqrt{2\sigma_{pp}%
}]$. Thus, the knowledge of these two particular intersections determine the
variances $\sigma_{xx}$ and $\sigma_{pp}$. In the general case $\mu\nu\neq0$
the intersection $\Omega\cap\ell_{\mu,\nu}$ can be parametrized by $x$ (or
$p)$ and is the interval defined by%
\[
\left(  \sigma_{pp}+\frac{2\mu}{\nu}\sigma_{xp}+\frac{\mu^{2}}{\nu^{2}}%
\sigma_{xx}\right)  x^{2}\leq\tfrac{1}{2}\hbar^{2}.
\]
We now observe that the coefficient of $x^{2}$ in this formula is exactly the
second denominator in formula (\ref{ab}) describing the Radon transform of the
Gaussian (\ref{Gauss1}). This is not a mere coincidence. We begin by noting
that the marginal conditions (\ref{marginal}) can be viewed as line integrals%
\begin{align*}
\int_{\ell_{0,1}}W\psi(x,p)dp  &  =|\psi(x)|^{2}\\
\int_{\ell_{1,0}}W\psi(x,p)dx  &  =|\widehat{\psi}(p)|^{2}.
\end{align*}
In the general case:

\begin{proposition}
Let $\psi,\widehat{\psi}\in L^{1}(\mathbb{R})\cap L^{2}(\mathbb{R})$. Let
$\ell_{\mu,\nu}^{X}$ be the straight line in $\mathbb{R}^{2}$ with equation
$\mu x+\nu p=X$. The Radon transform of $\psi$is the line integral
\begin{equation}
R_{\widehat{\rho}}(X,\mu,\nu)=\int_{\ell_{\mu,\nu}^{X}}W\psi(z_{_{\mu,\nu}%
})dz_{_{\mu,\nu}}\label{RadonGeom}%
\end{equation}
where $dz_{_{\mu,\nu}}$ is the push-forward of of the Lebesgue measure $dx$.
\end{proposition}

\begin{proof}
Denote by $W$ the integral on the left-hand side. Parametrizing $\ell_{\mu
,\nu}^{X}$ by $x=\nu t+\mu^{-1}X$, $p=-\mu t$ we have%
\begin{equation}
W=\lambda\int_{-\infty}^{\infty}W\psi(\nu t+\mu^{-1}X,-\mu t)dt\label{para}%
\end{equation}
where $\lambda=\sqrt{\mu^{2}+\nu^{2}}$. Returning to formula (\ref{dirac}), we
have%
\begin{align}
R_{\widehat{\rho}}(X,\mu,\nu) &  =\lambda^{-1}\int W\psi(U_{\mu,\nu}%
^{-1}(u,v))\delta(\lambda^{1}X-u)dudv\\
&  =\int W\psi(U_{\mu,\nu}^{-1}(\lambda^{-1}X,v))dv
\end{align}
that is, since $U_{\mu,\nu}^{-1}=%
\begin{pmatrix}
\mu/\lambda & -\nu/\lambda\\
\nu/\lambda & \mu/\lambda
\end{pmatrix}
$, and replacing $v$ with $s$,
\begin{equation}
R_{\widehat{\rho}}(X,\mu,\nu)=\int W\psi(\mu\lambda^{-2}X-v\lambda
^{-1}s,v\lambda^{-2}X+\mu\lambda^{-1}s))ds.\label{parabis}%
\end{equation}

\end{proof}

\section{Extension to higher dimensions: hints}

All of the above can be generalized without major difficulties to the case of
states in $L^{2}(\mathbb{R}^{n})$. We only sketch the main modifications
here,a detailed account will appear elsewhere. 

The definition (\ref{1}) of the Radon transform should be replaced with%
\begin{equation}
R_{\widehat{\rho}}(X,A,B)=\int_{\mathbb{R}^{n}\times\mathbb{R}^{n}%
}W(x,p)\delta(X-Ax-Bp)dpdx\label{radonn}%
\end{equation}
where $X\in\mathbb{R}^{n}$ and $A,B$ are two real square $n\times n$ matrices
with such that $A^{T}B=BA^{T}$ and $\operatorname*{rank}(A,B)=n$. This ensures
us that the subspace $\ell_{A,B}^{X}=\{(x,p):Ax+Bx=X\}$ of  $\mathbb{R}%
^{2n}\equiv T^{\ast}\mathbb{R}^{n}$ equipped with its standard symplectic
structure is an affine Lagrangian plane \cite{Birk}, which allows us to
rewrite the geometric version (\ref{RadonGeom}) as a surface integral
\[
R_{\widehat{\rho}}(X,A,B)=\int_{\ell_{A,B}^{X}}W\psi(z_{_{A,B}})dz_{_{A,B}}%
\]
(the fact that Lagrangian planes intervene is crucial since the unitary group
acts transitively on the Lagrangian Grassmannian). The inversion formula
(\ref{inverse}) should then be replaced with an expression of the type%
\begin{equation}
W\psi(x,p)=\left(  \frac{1}{2\pi\hbar}\right)  ^{n}\int R_{\widehat{\rho}%
}(X,A,B)e^{\frac{i}{\hslash}(X-Ax-Bp)}dXdAdB.\label{inversen}%
\end{equation}

\section*{APPENDIX A: The metaplectic group $\operatorname*{Mp}(n)$}

For a detailed study of the metaplectic group $\operatorname*{Mp}(n)$ see
\cite{Birk}.

Let $S=%
\begin{pmatrix}
A & B\\
C & D
\end{pmatrix}
$ be a real $2n\times2n$ matrix, where the \textquotedblleft
blocks\textquotedblright\ $A,B,C,D$ are $n\times n$ matrices. Let $J=%
\begin{pmatrix}
0 & I\\
-I & 0
\end{pmatrix}
$ the standard symplectic matrix. We have $S\in\operatorname*{Sp}(n)$ (the
symplectic group) if and only $SJS^{T}=S^{T}JS=J$. These relations are
equivalent to any of the two sets of conditions
\begin{align}
A^{T}C\text{, }B^{T}D\text{ \ \textit{are symmetric, and} }A^{T}D-C^{T}B  &
=I\label{cond12}\\
AB^{T}\text{, }CD^{T}\text{ \ \textit{are\ symmetric, and} }AD^{T}-BC^{T}  &
=I\text{.} \label{cond22}%
\end{align}
One says that $S$ is a \emph{free symplectic matrix }if $B$ is invertible,
i.e. $\det B\neq0$. To a free symplectic matrix is associated a generating
function: it is the quadratic form%
\begin{equation}
\mathcal{A}(x,x^{\prime})=\frac{1}{2}DB^{-1}x\cdot x-B^{-1}x\cdot x^{\prime
}+\frac{1}{2}B^{-1}Ax^{\prime}\cdot x^{\prime}. \label{wfree}%
\end{equation}
The terminology comes from the fact that the knowledge of $\mathcal{A}%
(x,x^{\prime})$ uniquely determines the free symplectic matrix $S$: we have%
\[%
\begin{pmatrix}
x\\
p
\end{pmatrix}
=%
\begin{pmatrix}
A & B\\
C & D
\end{pmatrix}%
\begin{pmatrix}
x^{\prime}\\
p^{\prime}%
\end{pmatrix}
\Longleftrightarrow\left\{
\begin{array}
[c]{c}%
p=\nabla_{x}\mathcal{A}(x,x^{\prime})\\
p^{\prime}=-\nabla_{x^{\prime}}\mathcal{A}(x,x^{\prime})
\end{array}
\right.
\]
as can be verified by a direct calculation.

Now, to every free symplectic matrix $S_{\mathcal{A}}$ we associate two
operators $\widehat{S}_{\mathcal{A},m}$ by the formula%
\begin{equation}
\widehat{S}_{\mathcal{A},m}\psi(x)=\left(  \tfrac{1}{2\pi\hbar}\right)
^{n/2}i^{m-n/2}\sqrt{|\det B^{-1}|}\int e^{\frac{i}{\hbar}\mathcal{A}%
(x,x^{\prime})}\psi(x^{\prime})d^{n}x^{\prime} \label{qft1}%
\end{equation}
where $m$ corresponds to a choice of argument for $\det B^{-1}$: $m=0$
$\operatorname{mod}2$ if $\det B^{-1}>0$ and $m=1$ $\operatorname{mod}2$ if
$\det B^{-1}<0$. It is not difficult to prove that the generalized Fourier
transforms $\widehat{S}_{\mathcal{A},m}$ are unitary operators on
$L^{2}(\mathbb{R}^{n})$. These operators generate the metaplectic group
$\operatorname*{Mp}(n)$. One shows that every $\widehat{S}\in
\operatorname*{Mp}(n)$ can be written (non uniquely) as a product
$\widehat{S}_{\mathcal{A},m}\widehat{S}_{\mathcal{A}^{\prime},m^{\prime}}$.
This group is a double covering of $\operatorname*{Sp}(n)$, the covering
projection being simply defined by
\begin{equation}
\pi_{\operatorname*{Mp}}:\operatorname*{Mp}(n)\longrightarrow
\operatorname*{Sp}(n)\text{ \ , \ }\pi_{\operatorname*{Mp}}(\widehat{S}%
_{\mathcal{A},m})=S_{\mathcal{A}}. \label{pimp}%
\end{equation}


\begin{thebibliography}{99}                                                                                               %


\bibitem {Asorey}M. Asorey, P. Facchi, V. I. Man'ko, G.\ Marmo, S. Pascazio,
and E. C. G. Sudarshan. Generalized quantum tomographic maps. \textit{Phys.
Scr}. 85065001 (2012)

\bibitem {Facchi}P. Facchi, Ligabo, and S. Pascazio. On the inversion of the
Radon transform: standard versus M 2 approach. \textit{J. Mod. Opt}. 57(3),
239--243 (2010)

\bibitem {Folland}G. B. Folland, \textit{Harmonic Analysis in Phase space}.
Annals of Mathematics studies, Princeton University Press, Princeton, N.J. (1989)

\bibitem {dariano}G.M. D'Ariano,\ Universal quantum observables, \textit{Phys.
Lett. A} 300, 1--6 (2002)

\bibitem {damapa}G.M. D'Ariano,\ C. Macchiavello, and M.G.A. Paris, Detection
of the density matrix through optical homodyne tomography without filtered
back projection, \textit{Phys. Rev. A} 50(5), 4298--4303 (1994)

\bibitem {Birk}M. de Gosson, \textit{Symplectic Geometry and Quantum
Mechanics}. Birkh\"{a}user, Basel, series \textquotedblleft Operator Theory:
Advances and Applications\textquotedblright, Vol. 166, 2006

\bibitem {go09}M. de Gosson. The Symplectic Camel and the Uncertainty
Principle: The Tip of an Iceberg? \textit{Found. Phys}. 99, 194 (2009)M. de
Gosson, \textit{The Wigner Transform, World Scientific}, series Advanced Texts
in Mathematics, 2017

\bibitem {Wigner}M. de Gosson, \textit{The Wigner Transform, World
Scientific}, series Advanced Texts in Mathematics, 2017

\bibitem {QHA}M. de Gosson. \textit{Quantum\ Harmonic Analysis: An
Introduction}. An Introduction. Volume 4 in the series Advances in Analysis
and Geometry. De Gruyter 2021

\bibitem {mdpi}M. A. de Gosson. The Pauli Problem for Gaussian Quantum States:
Geometric Interpretation. Mathematics 9(20), 2578 (2021)

\bibitem {Ibort}A. Ibort, V. I. Man'ko, G. Marmo, A. Simoni, and F.
Ventriglia. An introduction to the tomographic picture of quantum mechanics.
\textit{Phys. Scr}. 79, 065013 (2009)

\bibitem {Ibortbis}A. Ibort, V. I. Man'ko, G. Marmo, A. Simoni, and F.
Ventriglia. On the tomographic picture of quantum mechanics. \textit{Phys.
Lett. A} 374, 2614--2617 (2010)

\bibitem {Littlejohn}R. G. Littlejohn. The semiclassical evolution of wave
packets, \textit{Phys. Rep.} 138, 4-5 193--291 (1986)

\bibitem {Pauli}W. Pauli. \textit{General principles of quantum mechanics},
Springer Science \& Business Media, 2012 [original title: \textit{Prinzipien
der Quantentheorie}, publ. in : Handbuch der Physik, v.5.1, 1958]
\end{thebibliography}
\end{document}